%% file: main.tex
\documentclass[a4paper,onecolumn,11pt,unpublished]{quantumarticle}
\pdfoutput=1

\usepackage[utf8]{inputenc}
\usepackage[english]{babel}
\usepackage[T1]{fontenc}
\usepackage{amsmath}
\usepackage{amsthm}
\usepackage{hyperref}
\usepackage{amssymb}
\usepackage{mathtools}
\usepackage{booktabs}
\usepackage{array}
\usepackage{multirow}
\usepackage{placeins}

\usepackage{tikz}
\usepackage{lipsum}
\usepackage{braket}
\usepackage[numbers]{natbib}

\newcommand{\Tr}{\operatorname{Tr}}

\newcommand{\1}{\mathbf 1}
\newcommand{\D}{\mathcal D}

\newcommand{\e}{\mathrm e}

\newtheorem{theorem}{Theorem}[section]
\newtheorem{lemma}[theorem]{Lemma}
\newtheorem{corollary}[theorem]{Corollary}
\newtheorem{proposition}[theorem]{Proposition}
\newtheorem{remark}[theorem]{Remark}
\newtheorem*{theorem*}{Theorem}
\newtheorem*{proposition*}{Proposition}

\begin{document}
	
	\title{Near-Optimal Bounds on the Density of Low-Energy States of $k$-Local Hamiltonians and Faster Quantum Algorithms}
		
		\author{Sevag Gharibian}
		\affiliation{Department of Computer Science and Institute for Photonic Quantum Systems, Paderborn University, Germany.}
		
		\author{François Le Gall}
		\affiliation{Graduate School of Mathematics, Nagoya University, Japan.}
		
		\author{Ranitha Mataraarachchi}
		\affiliation{Graduate School of Mathematics, Nagoya University, Japan.}
		
		\orcid{0009-0003-6053-6667}
		\author{Suguru Tamaki}
		\affiliation{Graduate School of Information Science, University of Hyogo, Japan.}
	\maketitle
	
	\begin{abstract}
		
        Low-energy estimation and state preparation for general $k$-local Hamiltonians are fundamental challenges in quantum complexity theory. Buhrman et al.~ [BGLGST, PRL 2025] recently broke the natural Grover bound $O^\ast(2^{n/2})$ for both problems, with the improvement depending on the relative accuracy $\varepsilon$ and the locality $k$. In this work, we present faster exponential quantum algorithms for these problems, where the binary entropy function governs the runtime exponent. For sufficiently small $\varepsilon/k$, our algorithms improve the exponent by a factor of $\log(k/\varepsilon)$ over [BGLGST, PRL 2025]. Our main technical result is an entropy-governed lower bound on the dimension of the Hamiltonian's low-energy subspace, obtained by depolarizing its ground state. For fixed $k$, this bound is optimal up to constant factors in the exponent. The same framework yields tighter bounds for Heisenberg, $XY$, and Ising models on arbitrary interaction graphs.

	\end{abstract}

\input{introduction}
	\input{preliminaries}
	\input{proof}
	\input{algorithm}
	\input{conclusion}

    \section*{Disclosure of AI Tools}
    During the preparation of this work, the authors utilized ChatGPT (versions 5.4–5.6) and Gemini 3.1 as research assistants. Specifically, the models were employed to explore proof strategies, identify relevant literature~\cite{King2002}, and propose the complete proof for Theorem~\ref{th:main} based on that literature.

    ChatGPT-6 Astra was used to help draft and improve the presentation of this manuscript. The authors take full accountability for all content.

	\section*{Acknowledgments}
	SG acknowledges support from the Deutsche Forschungsgemeinschaft (DFG), projects 563388236 (Bridge-QS, SPP 2514) and 572703436 (QPUP), EU QuantERA/DFG project 583918116 (SDPCODE), and the BMFTR (PhoQuant).
	ST is supported by JSPS KAKENHI grant JP22K11909. FLG and RM are supported by JSPS KAKENHI grant Nos.~24H00071, 25K24674 and 25K24465, MEXT Q-LEAP grant No.~JPMXS0120319794, JST ASPIRE grant No.~JPMJAP2302 and JST CREST grant No.~JPMJCR24I4.
	\bibliographystyle{quantum}
	\bibliography{ref}
	
	\clearpage
	\appendix
	\input{appendix}
\end{document}

%% file: introduction.tex
\section{Introduction}\label{sec:intro}
\subsection{Background and Motivation}

Determining the low-energy properties of quantum many-body systems lies at the intersection of quantum physics and theoretical computer science. Specifically, the dynamics of a physical system are governed by its underlying Hamiltonian, whose low-energy spectrum determines many of the system's equilibrium properties at low temperatures. In the zero-temperature limit, an equilibrium state is supported on the \emph{ground-state subspace} (ground space) of the Hamiltonian. Two important computational tasks in this regime are:
\begin{enumerate}
	\item \emph{Ground-energy estimation}: estimating the ground-state energy of the Hamiltonian,
	\item \emph{Low-energy state preparation}: preparing a quantum state whose expected energy is close to the ground-state energy.
\end{enumerate}
These tasks have applications in quantum chemistry, condensed-matter physics, and quantum information.

Formally, ground-state energy estimation generalizes the NP-complete problem of Boolean constraint satisfaction. Specifically, consider \(n\)-qubit \(k\)-local Hamiltonian
\begin{equation}
	H\coloneqq\sum_{\alpha=1}^{m}h_\alpha,
\end{equation}
where each \(h_\alpha\) is a Hermitian operator acting nontrivially on at most \(k\) qubits, and \(k\) is a constant independent of \(n\). Each local interaction term \(h_\alpha\) generalizes the notion of a local clause in the study of classical Boolean constraint satisfaction. The \emph{ground-state energy} is $\lambda_0(H)$, for $
\lambda_0\leq\lambda_1\leq\cdots\leq\lambda_{2^n-1}$ the eigenvalues of $H$ (counted with multiplicity). Then, given $H$ and thresholds $a<b$, the latter separated by a
\emph{promise gap}
\begin{equation}\label{eq:gap}
	b-a\geq\frac{1}{\operatorname{poly}(n)},
\end{equation}
the \emph{\(k\)-local Hamiltonian problem} asks to decide whether
$
\lambda_0\leq a$ or $\lambda_0\geq b$,
assuming one of these cases holds. Analogous to NP-completeness if the Boolean satisfiability problem~\cite{cook,levin}, $k$-local Hamiltonian problem is \(\mathsf{QMA}\)-complete~\cite{kitaev,2LH} already for $k=2$. 
For these reasons, the $k$-local Hamiltonian problem is considered a benchmark problem in the study of \emph{quantum Hamiltonian complexity} (see~\cite{QHC} for a survey).

\paragraph{Setup for this work.} Let $M\coloneqq\sum_{a=1}^{m}\lVert h_a\rVert,$ denote the \emph{total interaction strength} of $H$, 
for $\lVert\cdot\rVert$ the spectral norm. We are interested in \emph{relative error approximations} $\varepsilon$, i.e. outputting estimate $\widehat{E}$ satisfying
$\widehat{E} \in [\lambda_0,\ \lambda_0+\varepsilon M]$.
Similarly, for low-energy state preparation, we wish to prepare $n$-qubit state $\rho$ satisfying
$
\operatorname{Tr}(H\rho)\leq\lambda_0+\varepsilon M.
$
Throughout, we assume $0<\varepsilon<1$ and $\varepsilon=\Omega(1/\operatorname{poly}(n))$, which includes the important case of constant $\varepsilon$ independent of $n$, the latter being closely related to the Quantum PCP conjecture~\cite{pcp,bh}.

Note the pursuit of relative approximations is relevant, as \(\mathsf{QMA}\)-completeness of ground-state energy approximation rules out (assuming $\mathsf{BQP}\neq\mathsf{QMA}$) \emph{additive} error inverse polynomial approximations, i.e. within error $\pm \epsilon$ instead of $\pm \epsilon M$. For the same reason, ground state preparation tasks are expected to be intractable. 
In fact, 
recently, Chia, Hasegawa, Le Gall, and Shen~\cite{CHLS26} showed that, assuming the Quantum Strong Exponential-Time Hypothesis (QSETH), ground-energy estimation for \(3\)-local Hamiltonians with \emph{constant additive error} cannot be performed in time
$
O^*\left(2^{(1-\xi)n/2}\right)
$
for any \emph{constant} \(\xi>0\), where \(O^*(\cdot)\) suppresses polynomial factors in \(n\). Thus, assuming QSETH, obtaining a constant-factor reduction in the exponent \(n/2\) requires relaxing constant additive accuracy. This motivates our pursuit of additive error growing with \(M\).

\subsection{Previous Results}
Without a nontrivial lower bound on the dimension of the target low-energy subspace, the most straightforward quantum algorithm has a worst-case running time of \(O^*(2^{n/2})\). This approach combines quantum phase estimation (QPE)~\cite{qpe} with amplitude amplification~\cite{ampamp}. More precisely, if the relevant low-energy subspace has dimension \(D\), applying QPE to the maximally mixed state\footnote{Practically, we have to apply QPE to a purification of the maximally mixed state.} finds an eigenstate in this subspace with probability \(D/2^n\). Amplitude amplification increases this probability to a constant using \(O\bigl(\sqrt{2^n/D}\bigr)\) applications of the procedure. Without additional information about the low-energy spectrum, one can use only the trivial bound \(D\geq 1\), resulting in the \(O^*(2^{n/2})\) runtime.

Recently, Buhrman et al.~\cite{beat} broke this \(O^*(2^{n/2})\) bound by establishing a nontrivial lower bound on the number of low-energy eigenstates of any $k$-local Hamiltonian, obtaining the first quantum algorithms for ground-energy estimation and low-energy state preparation strictly faster than \(O^*(2^{n/2})\).

\begin{theorem}[Buhrman et al.~\cite{beat}]\label{th:buhrman}
	Let $H$ be an $n$-qubit $k$-local Hamiltonian. For any $\varepsilon =\Omega(1/\mathrm{poly}(n))$, there exists a quantum algorithm that, with high probability, prepares an $n$-qubit quantum state $\rho$ satisfying
	$ \operatorname{Tr}[H\rho] \le \lambda_0+\varepsilon M $
	and outputs a corresponding energy estimate $\widehat{E}$ lying within the window
	$ \widehat{E} \in [\lambda_0,\ \lambda_0+\varepsilon M]. $
	The running time of the algorithm is
	$$O^*\left( 2^{\frac{n}{2}\left(1 - \frac{\varepsilon}{2k+\varepsilon}\right)}\right).\footnote{In the case of low-energy estimation, \cite{beat} also shows how to improve the term $\frac{\varepsilon}{2k+\varepsilon}$ to $\frac{\varepsilon}{k+\varepsilon}$, at the price of slightly increasing the window to $[\lambda_0-\varepsilon M,\ \lambda_0+\varepsilon M]$,}$$
	where $O^*$ suppresses polynomial factors in $n.$ 
\end{theorem}

\noindent Subsequent work by Buhrman et al.~\cite{classicalMaxKSAT} obtained a classical algorithm with an \emph{entropy-governed} speedup for \emph{classical} constraint optimization problems (e.g., MAX-$k$-SAT). For such problems, the associated Hamiltonian and ground state are classical. Most recently, Mataraarachchi, Le Gall, and Tamaki~\cite{MGT26} showed that an entropy-governed runtime for quantum algorithms can also be obtained for low-energy estimation and state preparation of \emph{quantum} $k$-local Hamiltonians \emph{if} one restricts to estimating the smallest energy achievable by depth-bounded quantum states (defined shortly in Equation (\ref{eq:Ed})). However, whether the same could be said for the true minimum energy, $\lambda_0$, remained open.

\subsection{Our Contributions}
We resolve this question by giving an entropy-governed speedup for low-energy estimation and state preparation for any quantum $k$-local Hamiltonian. Thus, our algorithms directly improve upon the runtimes established in~\cite{beat}.

\begin{theorem}[Our Results]
	\label{th:main}
	Let \(H\) be an \(n\)-qubit \(k\)-local Hamiltonian. For any
	\(\varepsilon=\Omega(1/\operatorname{poly}(n))\), there exists a quantum algorithm that, with high probability, prepares an \(n\)-qubit quantum state \(\rho\) satisfying
$	
	\operatorname{Tr}[H\rho]
	\leq
	\lambda_0+\varepsilon M
	$
	and outputs a corresponding energy estimate \(\widehat{E}\) lying within the window
	$
	\widehat{E}
	\in
	[\lambda_0,\lambda_0+\varepsilon M]$.
	The running time of the algorithm is
	\[
	O^*\left(
	2^{\frac{n}{2}\left(
		1-\frac{3}{5}
		h_2\left(\frac{\varepsilon}{6k}\right)
		\right)}
	\right),
	\]
	where $h_2(\cdot)$ denotes the binary entropy function in bits.
\end{theorem}
\noindent The lower bound on the low-energy subspace dimension underlying Theorem~\ref{th:main}, specifically $\Omega^\ast(2^{\frac35 n h_2\left(\frac{\mu}{6k}\right)})$ (Corollary \ref{cor:coarse}; $\Omega^*$ suppresses inverse-polynomial factors in $n$) is nearly optimal. Proposition~\ref{prop:tightness} exhibits a classical diagonal Hamiltonian with the same binary-entropy scaling for fixed \(k\), up to constant factors in the exponent (specifically, $2^{\Theta\left(
	nh_2\left(\frac{\mu}{6k}\right)
	\right)}$). Thus, although the numerical constants in our bound may be sharpened, its dependence on the binary entropy cannot be improved in general.

We also consider specialized cases by restricting the types of local interactions while allowing for arbitrary interaction graphs. Consequently, we obtain tighter bounds for the Heisenberg, $XY$, and Ising models, which are summarized in Table~\ref{tab:response}.

For constant \(k\) and \(\varepsilon\), the exponent in Theorem~\ref{th:main} is strictly smaller than \(n/2\), and hence our algorithms beat the natural \(O^*(2^{n/2})\) bound. Further, for small \(\varepsilon/k\),
\[
\frac{3}{5}h_2\left(\frac{\varepsilon}{6k}\right)
=
\Theta\left(
\frac{\varepsilon}{k}
\log\frac{k}{\varepsilon}
\right),
\]
whereas the corresponding term in the exponent of Buhrman et al.~\cite{beat} satisfies
\[
\frac{\varepsilon}{2k+\varepsilon}
=
\Theta\left(\frac{\varepsilon}{k}\right).
\]
Thus, our algorithms achieve a logarithmic factor speedup in \(k/\varepsilon\) in the runtime exponent over~\cite{beat}. See Table~\ref{tab:comparison} for a numerical comparison of runtimes.

	\begin{table}[h]
		\centering
		\renewcommand{\arraystretch}{1.3}
		\begin{tabular}{cc cc}
			\toprule
			& & \multicolumn{2}{c}{\textbf{Time Complexity Exponent Coefficient ($c$)}} \\
			\cmidrule(lr){3-4} 
			$\boldsymbol{k}$ & $\boldsymbol{\varepsilon}$ & \textbf{Buhrman et al.~\cite{beat}} & \textbf{Our Work} \\
			\midrule
			\multirow{4}{*}{$2$} 
			& $1/8  $ & $0.4848485$ & $0.4749371$ \\
            & $0.05 $ & $0.4938272$ & $0.4883168$ \\
            & $0.01 $ & $0.4987531$ & $0.4970823$ \\
            & $0.001$ & $0.4998750$ & $0.4996252$ \\
			\midrule
			\multirow{4}{*}{$3$} 
			& $1/8  $ & $0.4897959$ & $0.4820675$ \\
            & $0.05 $ & $0.4958678$ & $0.4917229$ \\
            & $0.01 $ & $0.4991681$ & $0.4979573$ \\
            & $0.001$ & $0.4999167$ & $0.4997404$ \\
			\midrule
			\multirow{4}{*}{$10$} 
			& $1/8  $ & $0.4968944$ & $0.4935324$ \\
            & $0.05 $ & $0.4987531$ & $0.4970823$ \\
            & $0.01 $ & $0.4997501$ & $0.4993003$ \\
            & $0.001$ & $0.4999750$ & $0.4999134$ \\
			\bottomrule
		\end{tabular}
		\vspace{0.2cm}
		\caption{Comparison of the time complexity exponent $c$, where algorithms run in time $O^*(2^{cn})$. A smaller $c$ indicates a faster execution. Our algorithms are faster than Buhrman et al.~\cite{beat} for all $k$ and $\varepsilon$ values.}
		\label{tab:comparison}
	\end{table}

As for a comparison with the algorithms of~\cite{MGT26}, the latter targets the minimum energy \(E_d\) achievable by a depth-\(d\) state rather than \(\lambda_0\). i.e.,
\begin{equation}\label{eq:Ed}
	E_d
	\coloneqq
	\min_{U_d\in\mathcal{C}_d}
	\bra{0}^{\otimes n}U_d^\dagger H U_d\ket{0}^{\otimes n},
\end{equation}
where \(d\) is a constant independent of \(n\), and \(\mathcal{C}_{d}\) denotes the set of \(n\)-qubit quantum circuits of depth \(d\), where each layer consists of two-qubit gates and arbitrary single-qubit gates. Under this convention, $E_0$ is the minimum energy due to a product state. For every \(d\geq1\), our algorithms in Theorem~\ref{th:main} are faster than those of~\cite{MGT26} and target the true ground state energy \(\lambda_0\leq E_d\). For \(d=0\), the algorithms of~\cite{MGT26} are faster but generally target an energy higher than \(\lambda_0\). Thus, the algorithms in~\cite{MGT26} are faster with the same energy guarantee only when the Hamiltonian has a product ground state.

\subsection{Proof Overview}\label{subsec:proof-overview}

Our proof is motivated by approaches in \cite{classicalMaxKSAT} and \cite{MGT26}, where an optimal assignment is perturbed locally, and then the resulting family of nearly optimal assignments is counted. For classical MAX-\(k\)-SAT, flipping suitable bits of an optimal assignment produces many distinct near-optimal assignments that can be counted directly~\cite{classicalMaxKSAT}. Similarly, for \emph{product} ground states, suitable local qubit flips produce mutually orthogonal low-energy states~\cite{MGT26}. However, for an arbitrary entangled ground state, different local perturbations need not produce mutually orthogonal states, so a conclusion regarding the dimension of the low-energy subspace becomes more challenging.

In this work, we overcome this difficulty by replacing direct counting with a probabilistic argument based on depolarization. The support pattern of a local perturbation can be represented by a bit string in the Hamming cube \(\{0,1\}^n\), where a \(1\) indicates that a qubit is perturbed. Independently choosing which qubits to perturb within a fixed set produces a Hamming subcube. Depolarization instead samples perturbations across all qubits while controlling their typical Hamming weight, giving a picture closer to a Hamming ball. Rather than counting the perturbed states individually, we analyze the entropy of the resulting mixture. Specifically, in Section~\ref{sec:proof}, we apply a depolarizing channel \(\D_b\) independently to each qubit of an arbitrary ground-state density operator \(\omega\), obtaining a state \(\rho_b\).

Next, two properties of the depolarized state $\rho_b$ are established. First, $k$-locality of $H$ ensures that depolarization increases its expected energy $\operatorname{Tr}(H\rho_b)$ by only a controlled amount, proportional to the total interaction strength $M$. Second, using a theorem in~\cite{King2002}, we bound the Schatten \(p\)-th moment of \(\rho_b\) in terms of the binary Rényi \(p\)-entropy \(H_p(b)\). Thus, \(\rho_b\) remains close to the ground-state energy while having a small \(p\)-th moment controlled by \(H_p(b)\).

Then, these two properties are combined to obtain a lower bound on the number of low-energy eigenstates as follows. Let \(N_H(E)\) denote the number of eigenvalues of \(H\), counted with multiplicity, that are at most \(E\). The key is to focus on the overlap $\Tr(P\rho_b)$ for $P$ the projector onto the low energy space of $H$. Since $\rho_b$ has low energy, one first has that $\Tr(P\rho_b)$ has a non-trivial lower bound. Second, $\Tr(P\rho_b)$ is upper bounded via H\"older inequality for the $p/(p-1)$- and $p$-Schatten norms, yielding an expression containing both \(N_H(E)\) and the Schatten \(p\)-th moment of \(\rho_b\), respectively. Section~\ref{sec:proof} thus gives the following density-of-states bound, expressed in terms of the binary entropy function \(h_2(\cdot)\):
\begin{equation}\label{eqn:LB}
N_H(\lambda_0+\mu M)
=
\Omega^*\left(
2^{\frac35 n h_2\left(\frac{\mu}{6k}\right)}
\right),
\end{equation}
where $\Omega^*$ suppresses inverse-polynomial factors in $n$.

Finally, in Section~\ref{sec:algorithmic-implications}, we use the above bound to obtain entropy-governed quantum algorithms. Directly applying the algorithmic framework of~\cite{MGT26} gives Theorem~\ref{th:main}.

\subsection{Other Prior Works}
By restricting the interaction graph of a $k$-local Hamiltonian, the ground-state energy can in some cases be classically efficiently approximated to \emph{any} fixed relative accuracy $\epsilon$, yielding a Polynomial-Time Approximation Scheme (PTAS). This is the case, for example, for planar and dense interaction graphs~\cite{bravyi,GK12,bh,bergamaschi:LIPIcs.ICALP.2023.20}. These algorithms, however, require said geometric structure, and do not directly extend to general interaction graphs. This work imposes no geometric restrictions and allows long-range and all-to-all interactions.

Recently, Stacey Jeffery and Freek Witteveen~\cite{guiding1} obtained optimal quantum algorithms for additive error ground-state energy estimation given a guiding state, with Rolando D. Somma and Ronald de Wolf~\cite{guiding2} establishing matching lower bounds. These algorithms are polynomial-time, but require a guiding state as input (i.e. a non-trivial approximation to the ground state). Our algorithms do not require a guiding state.

Variational methods, most notably the Variational Quantum Eigensolver (VQE), provide a heuristic approach to low-energy estimation and state preparation~\cite{vqe}. However, their performance depends strongly on the choice of ansatz and on the resulting optimization landscape. VQE optimization landscapes may exhibit barren plateaus~\cite{barren}, and the classical optimization problem underlying variational quantum algorithms is $\mathsf{QCMA}$-complete in the worst case~\cite{vqeNP,BGK23}. Consequently, such methods do not provide worst-case convergence guarantees for arbitrary local Hamiltonians. In contrast, the objective of this work is to obtain rigorous worst-case guarantees on both the energy achieved and the running time.

\paragraph{Organization of the paper.}
Section~\ref{sec:preliminaries} introduces the notation and analytical tools used throughout the paper. In Section~\ref{sec:proof}, we prove our density-of-states bound: Subsection~\ref{subsec:depolarizing-channel} bounds the Schatten \(p\)-th moment of the depolarized state, Subsection~\ref{subsec:energy-bound} bounds its expected energy, and Subsection~\ref{subsec:density-bound} combines these results to obtain the Rényi-parameterized bound and its binary-entropy corollary. Section~\ref{sec:algorithmic-implications} derives the resulting quantum algorithms for low-energy state preparation and estimation. This section also gives sharper runtime bounds for specific interaction models. Finally, Section~\ref{sec:conclusion} provides a brief conclusion and summarizes the broader implications of our work.

%% file: preliminaries.tex
\section{Preliminaries}\label{sec:preliminaries}

In this section, we introduce the notation and standard results used throughout the paper. We retain the notation introduced in Section~\ref{sec:intro}. In particular, we consider an $n$-qubit $k$-local Hamiltonian $H=\sum_{\alpha=1}^{m}h_\alpha$ with eigenvalues $\lambda_0 \le \lambda_1 \le \dots \le \lambda_{2^n-1}$ and total interaction strength $M = \sum_{\alpha=1}^m \Vert h_\alpha \Vert.$

A \emph{ground-state density operator} is any density operator \(\omega\) supported on the ground space of $H$. Equivalently,
\[
\operatorname{Tr}(H\omega)=\lambda_0.
\]
\noindent For an energy threshold \(E\in\mathbb{R}\), let
\[
P_{\leq E}\coloneqq\mathbf{1}[H\leq E]
\]
denote the \emph{spectral projector} onto the subspace spanned by eigenvectors of \(H\) with eigenvalues at most \(E\). We define the \emph{cumulative density of states} by
\[
N_H(E)
\coloneqq
\operatorname{Tr}(P_{\leq E})
=
\operatorname{rank}(P_{\leq E}).
\]
Thus, \(N_H(E)\) is the number of eigenvalues of \(H\) that are at most \(E\), counted with multiplicity.

Next, for Hermitian operators \(A\) and \(B\), we write
$
A\succeq B
$
when \(A-B\) is positive semidefinite. For \(p\geq1\), the Schatten \(p\)-norm of an operator \(A\) is
$
\lVert A\rVert_p
\coloneqq
\bigl(\operatorname{Tr}|A|^p\bigr)^{1/p}.
$
If \(p,q\geq1\) satisfy \(1/p+1/q=1\), \emph{Hölder's inequality} for Schatten norms states that
$
\left|\operatorname{Tr}(A^\dagger B)\right|
\leq
\lVert A\rVert_q\lVert B\rVert_p.
$

In this work, we use \emph{Rényi \(p\)-entropy}, which for \(p>1\) and density operator \(\rho\) is
\[
S_p(\rho)
\coloneqq
\frac{1}{1-p}\log\operatorname{Tr}(\rho^p).
\]
For \(0\leq b\leq1\), we define the \emph{binary} Rényi \(p\)-entropy by
\[
H_p(b)
\coloneqq
\frac{1}{1-p}
\log\bigl((1-b)^p+b^p\bigr).
\]
We also define the \emph{binary entropy} in natural units and in bits, respectively, by
\[
h(b)
=
-b\log b-(1-b)\log(1-b),
\qquad
h_2(b)
=
\frac{h(b)}{\log 2}.
\]
All logarithms are natural unless their base is stated explicitly. 

Finally, a \emph{quantum channel} is a completely positive and trace-preserving linear map. For a
quantum channel \(\Phi\), its \emph{maximal output Schatten \(p\)-norm} is
\[
\nu_p(\Phi)
\coloneqq
\max_{\rho}\lVert\Phi(\rho)\rVert_p,
\]
where the maximum is over all density operators on the input space. For completeness, the \emph{Pauli matrices} are
\[
X=
\begin{pmatrix}
	0&1\\
	1&0
\end{pmatrix},
\qquad
Y=
\begin{pmatrix}
	0&-\mathrm{i}\\
	\mathrm{i}&0
\end{pmatrix},
\qquad
Z=
\begin{pmatrix}
	1&0\\
	0&-1
\end{pmatrix}.
\]

For an operator $h$ supported on $A$ and $i\in A$, let 
\begin{equation}\label{eq:defR}
     R_i(h):=\frac{I_i}{2}\otimes\operatorname{Tr}_i(h)
 =\frac14\bigl(h+X_i hX_i+Y_i hY_i+Z_i hZ_i\bigr),
\end{equation}
where the identity factor is placed on qubit $i$. For the given local decomposition $H=\sum_\alpha h_\alpha$, with supports $A_\alpha$ and $M=\sum_\alpha\|h_\alpha\|>0$, define \begin{equation}
 y(H):=\frac{2}{M}\sum_\alpha\sum_{i\in A_\alpha}
                  \|h_\alpha-R_i(h_\alpha)\|.
 \label{eq:interaction-response-parameter}
\end{equation}
Since $\|h-R_i(h)\|\le(3/2)\|h\|$, one has $0\le y(H)\le3k$. 

Every operator on a set $A$ of qubits has a unique expansion $\sum_Q c_Q Q$ over \emph{Pauli strings} $Q=\bigotimes_{j\in A}Q_j$, where $Q_j\in\{I,X,Y,Z\}$. Since $X$, $Y$ and $Z$ are traceless while $\operatorname{Tr}(I)=2$, the map $R_i$ retains exactly those Pauli strings that act as the identity on qubit $i$: if $h_\alpha=\sum_Q c_Q Q$, then
\begin{equation}\label{eq:Rsplit}
 R_i(h_\alpha)=\sum_{Q\,:\,Q_i=I}c_Q\,Q,
 \qquad
 h_\alpha-R_i(h_\alpha)=\sum_{Q\,:\,Q_i\neq I}c_Q\,Q .
\end{equation}
Hence $\|h_\alpha-R_i(h_\alpha)\|$ measures the part of the interaction that acts nontrivially on qubit $i$, and equals $\|h_\alpha\|$ whenever every Pauli string of $h_\alpha$ does so.

For asymptotic scaling, we use the notation $f(n) = O^\ast(g(n))$ to indicate that $f(n) \le g(n) \cdot \mathrm{poly}(n)$, and $f(n) = \Omega^\ast(g(n))$ to indicate that $f(n) \ge g(n)\cdot\mathrm{poly}^{-1}(n)$, thereby suppressing polynomial factors in $n$ for both the upper and lower bounds, respectively.

%% file: proof.tex
\section{An Entropy-Governed Low-Energy Spectrum}\label{sec:proof}
In this section, we prove our entropy-governed density-of-states bound for arbitrary $k$-local Hamiltonians. We first introduce the qubit depolarizing channel in Subsection~\ref{subsec:depolarizing-channel} and bound the Schatten \(p\)-th moment of its \(n\)-qubit output. Then, in Subsection~\ref{subsec:energy-bound}, we bound the increase in energy caused by applying this channel to a ground-state density operator. Finally, in Subsection~\ref{subsec:density-bound}, we combine these two bounds to obtain the density-of-states bound that depends on the Rényi entropy of the depolarized state. To compare our bound with the previous results in \cite{beat} and \cite{MGT26}, we rewrite this density-of-states bound in terms of the binary entropy.

\subsection{The Qubit Depolarizing Channel}\label{subsec:depolarizing-channel}

In this subsection, we define our qubit depolarizing channel and prove some of its properties.
For $0\le b\le2/3$, define the qubit depolarizing channel by
\begin{equation}\label{eq:depol-definition}
	\D_b(\rho)=(1-2b)\rho+bI_2.
\end{equation}

We note that $\D_b$ is a completely positive and trace-preserving (CPTP) map for $0<b\leq 2/3$, and has the following alternative representation using Pauli matrices:
\begin{equation}\label{eq:pauli-representation}
	\D_b(\rho)
	=
	\left(1-\frac{3b}{2}\right)\rho
	+\frac b2\left(X\rho X+Y\rho Y+Z\rho Z\right).
\end{equation}
See Appendix~\ref{ap: properties} for the proofs of the above claims.

We now determine the spectrum of the output when $\D_b$ acts on a pure state.

\begin{lemma}\label{lem:qubit_eigenvalues}
	For $0\le b<1/2$, $\D_b$ maps every pure state to a state with eigenvalues $1-b$ and $b$.
\end{lemma}

\begin{proof}
	Let $\ket{\psi}\bra{\psi}$ be a single qubit pure state, and let $\ket{\psi^\perp}$ be an orthogonal unit vector. Since
	$
	I_2=\ket{\psi}\bra{\psi}+\ket{\psi^\perp}\bra{\psi^\perp},
	$
	we have
	$$
	\D_b(\ket{\psi}\bra{\psi})
	=(1-2b)\ket{\psi}\bra{\psi}+bI_2
	=(1-b)\ket{\psi}\bra{\psi}
	+b\ket{\psi^\perp}\bra{\psi^\perp}.
	$$
	Thus the eigenvalues are $1-b$ and $b$.
\end{proof}

The next lemma determines the maximal output $p$-norm of $\D_b$. i.e., $\nu_p(\D_b).$
		
\begin{lemma}\label{lem:single-qubit}
	For \(0\leq b<1/2\) and \(p>1\),
	\[
	\nu_p(\D_b)^p=(1-b)^p+b^p.
	\]
\end{lemma}

\begin{proof}
	Let
	$
	\rho=\sum_i\lambda_i\ket{\psi_i}\bra{\psi_i}
	$
	be a spectral decomposition of an arbitrary single-qubit state. By the linearity of \(\D_b\), followed by the triangle inequality,
	\[
	\|\D_b(\rho)\|_p
	\leq
	\sum_i\lambda_i
	\|\D_b(\ket{\psi_i}\bra{\psi_i})\|_p.
	\]
	By Lemma~\ref{lem:qubit_eigenvalues}, the output corresponding to every pure input has eigenvalues \(1-b\) and \(b\). Therefore,
	\[
	\|\D_b(\ket{\psi_i}\bra{\psi_i})\|_p
	=
	\bigl((1-b)^p+b^p\bigr)^{1/p}
	\]
	for every \(i\). It follows that
	\[
	\|\D_b(\rho)\|_p
	\leq
	\bigl((1-b)^p+b^p\bigr)^{1/p}.
	\]
	Since every pure input state attains equality, the claim follows.
\end{proof}

For an \(n\)-qubit density operator \(\omega\), let
\begin{equation}\label{eq:depol}
	\rho_b\coloneqq\mathcal D_b^{\otimes n}(\omega).
\end{equation}
To bound the \(p\)-th moment of the depolarized \(n\)-qubit state $\rho_b$, we use the following multiplicativity result of King~\cite{King2002} for maximal output \(p\)-norms.

\begin{proposition}[Theorem~3 of \cite{King2002}]
	\label{thm:king}
	Let \(\mathcal D_b\) be a qubit depolarizing channel and let
	\(\Phi\) be an arbitrary quantum channel. For every real \(p\geq1\),
	\[
	\nu_p(\mathcal D_b\otimes\Phi)
	=
	\nu_p(\mathcal D_b)\nu_p(\Phi).
	\]
\end{proposition}

Iterating Proposition~\ref{thm:king} over \(n\) tensor factors gives
\begin{equation}\label{eq:king-iteration}
	\nu_p\bigl(\mathcal D_b^{\otimes n}\bigr)
	=
	\nu_p(\mathcal D_b)^n.
\end{equation}

We now use Equation~\eqref{eq:king-iteration} to bound the \(p\)-th moment of \(\rho_b\), which equivalently gives a lower bound on its binary Rényi \(p\)-entropy.

\begin{lemma}\label{lem:moment}
	For every $\rho_b=\D_b^{\otimes n}(\omega)$, every $0\le b<1/2$, and every $p>1$,
	\[
	\Tr(\rho_b^p)
	\le
	\bigl((1-b)^p+b^p\bigr)^n
	=
	\exp\bigl(-(p-1)nH_p(b)\bigr),
	\]
	where
	\[
	H_p(b)=\frac{1}{1-p}\log\bigl((1-b)^p+b^p\bigr)
	\]
	is the binary R\'enyi $p$-entropy.
\end{lemma}

\begin{proof}
	By the definition of the Schatten $p$-norm,
	$
	\Tr(\rho_b^p)=\|\rho_b\|_p^p.
	$
	Since $\rho_b$ is an output of $\D_b^{\otimes n}$,
	\[
	\|\rho_b\|_p
	\le
	\nu_p(\D_b^{\otimes n}).
	\]
	By Eq.~\eqref{eq:king-iteration} and Lemma~\ref{lem:single-qubit},
	\[
	\nu_p(\D_b^{\otimes n})^p
	=
	\nu_p(\D_b)^{np}
	=
	\bigl((1-b)^p+b^p\bigr)^n.
	\]
	The final equality in the statement follows directly from the definition of $H_p(b)$.
\end{proof}

\subsection{Energy Bound for the Depolarized Ground State}\label{subsec:energy-bound}

We now specialize \(\omega\) to a ground-state density operator of $H$. Thus,
\(\operatorname{Tr}(H\omega)=\lambda_0\), and we recall the definition of \(\rho_b\) from Eq.~\eqref{eq:depol}. This subsection bounds the energy increase in the depolarized state $\rho_b$.

\begin{lemma}\label{lem:energy}
	For every $0\le b\le2/3$,
	\[
	0\le
	\Tr\bigl((H-\lambda_0I)\rho_b\bigr)
	\le
	3kbM.
	\]
\end{lemma}

\begin{proof}
	By Eq.~\eqref{eq:pauli-representation}, each local depolarizing channel is a random Pauli channel. Hence
	\[
	\rho_b=\sum_U P_U\,U\omega U^\dagger,
	\]
	where $U=U_1\otimes\cdots\otimes U_n$ ranges over tensor products of Pauli operators and $\{P_U\}$ is a probability distribution.
	
	Fix a local term $h_\alpha$ of $H$ with support $A_\alpha$ of size $s_\alpha\le k$. If $U$ is the identity on every qubit in $A_\alpha$, then $U^\dagger h_\alpha U=h_\alpha$. By Eq.~\eqref{eq:pauli-representation}, the identity is applied to a single qubit with probability $1-3b/2$. Thus $U$ is the identity on the whole support $A_\alpha$ with probability
	\[
	\left(1-\frac{3b}{2}\right)^{s_\alpha}.
	\]
	For every other $U$,
	\[
	\left|
	\Tr(h_\alpha U\omega U^\dagger)-\Tr(h_\alpha\omega)
	\right|
	\le
	\left|\Tr(h_\alpha U\omega U^\dagger)\right|
	+
	\left|\Tr(h_\alpha\omega)\right|
	\le
	2\|h_\alpha\|.
	\]
	Therefore,
	\[
	\left|
	\Tr(h_\alpha\rho_b)-\Tr(h_\alpha\omega)
	\right|
	\le
	2\left[1-\left(1-\frac{3b}{2}\right)^{s_\alpha}\right]\|h_\alpha\|.
	\]
	Since $s_\alpha\le k$ and $0\le1-3b/2\le1$, summing over all local terms gives
	\[
	\left|
	\Tr(H\rho_b)-\Tr(H\omega)
	\right|
	\le
	2\left[1-\left(1-\frac{3b}{2}\right)^k\right]M.
	\]
	Moreover, $H-\lambda_0I\succeq0$, so
	\[
	0\le\Tr\bigl((H-\lambda_0I)\rho_b\bigr).
	\]
	Finally, $1-(1-x)^k\le kx$ for $x\in[0,1]$. Taking $x=3b/2$ yields
	\[
	\Tr\bigl((H-\lambda_0I)\rho_b\bigr)
	\le3kbM.
	\]
\end{proof}

We now give an energy bound that depends on the type of interactions in $H$. Since $y\leq 3k$, the following bound never exceeds the bound in Lemma~\ref{lem:energy}.
\begin{lemma}
\label{cor:response-energy}
Let $H$ have $M>0$, and write $y=y(H)$ as defined in Eq.~\eqref{eq:interaction-response-parameter}.
For every $0\le b\le2/3$,
\[
 0\le\operatorname{Tr}[(H-\lambda_0I)\rho_b]\le ybM.
\]
\end{lemma}

\begin{proof}
Let $D_{b,i}$ act as $D_b$ on qubit $i$ and as the identity elsewhere. By Eq.~\eqref{eq:pauli-representation} and the definition of $R_i$ in Eq.~\eqref{eq:defR}, for every operator $A$,
\[
 D_{b,i}(A)-A
 =-\frac{3b}{2}A+\frac b2\sum_{P\in\{X,Y,Z\}}P_iAP_i
 =2b\bigl(R_i(A)-A\bigr).
\]
Moreover, since $0\le b\le2/3$ and each $P_i$ is unitary,
\begin{equation}\label{eq:normbound}
 \|D_{b,i}(A)\|
 \le\left(1-\frac{3b}{2}\right)\|A\|
    +\frac b2\sum_{P\in\{X,Y,Z\}}\|P_iAP_i\|
 =\|A\|.        
\end{equation}

Fix $\alpha$, write $A_\alpha=\{i_1,\ldots,i_s\}$, and set $F_0=\mathrm{id}$ and $F_j=D_{b,i_1}\circ\cdots\circ D_{b,i_j}$. Since $D_b(I)=I$, we have
$D_b^{\otimes n}(h_\alpha)=F_s(h_\alpha)$. Thus
\[
\begin{aligned}
 \|D_b^{\otimes n}(h_\alpha)-h_\alpha\|
 &=\left\|\sum_{j=1}^s
       \bigl(F_j(h_\alpha)-F_{j-1}(h_\alpha)\bigr)\right\|\\
 &\le\sum_{j=1}^s
       \bigl\|F_{j-1}\bigl(D_{b,i_j}(h_\alpha)-h_\alpha\bigr)\bigr\|\\
 &\le\sum_{j=1}^s\|D_{b,i_j}(h_\alpha)-h_\alpha\|\\
 &=2b\sum_{i\in A_\alpha}\|R_i(h_\alpha)-h_\alpha\|.
\end{aligned}
\]
Here the second inequality follows by applying the preceding norm bound in Eq.~\eqref{eq:normbound} to each channel in $F_{j-1}$. By linearity and the triangle inequality,
$$
\|D_b^{\otimes n}(H)-H\|
 \le\sum_\alpha\|D_b^{\otimes n}(h_\alpha)-h_\alpha\|
 \le2b\sum_\alpha\sum_{i\in A_\alpha}
          \|R_i(h_\alpha)-h_\alpha\|
 =ybM.
$$

Finally, cyclicity of the trace gives $\operatorname{Tr}(AP_iBP_i)=\operatorname{Tr}(P_iAP_iB)$. Together with Eq.~\eqref{eq:pauli-representation}, this yields $\operatorname{Tr}[A D_{b,i}(B)]
=\operatorname{Tr}[D_{b,i}(A)B]$, expressing the self-adjointness of $D_{b,i}$. Applying this identity on each qubit, and using $\rho_b=D_b^{\otimes n}(\omega)$ and $\operatorname{Tr}(H\omega)=\lambda_0$, we obtain
\[
 \operatorname{Tr}[(H-\lambda_0I)\rho_b]
 =\operatorname{Tr}\bigl[(D_b^{\otimes n}(H)-H)\omega\bigr]
 \le\|D_b^{\otimes n}(H)-H\|
 \le ybM,
\]
where we used that $\omega$ is a density operator. The lower bound follows from $H-\lambda_0I\succeq0$.
\end{proof}

\subsection{Lower Bounds on the Density of Low-Energy States}\label{subsec:density-bound}

This subsection establishes our main lower bound on the density of low-energy states. We first prove a general lemma (Lemma~\ref{lem:smoothing}) showing that a state with low expected energy and a small Schatten \(p\)-moment implies a lower bound on the dimension of the low-energy subspace. We then apply this lemma to the depolarized ground state, using Lemma~\ref{lem:energy} to control its energy and Lemma~\ref{lem:moment} to control its Schatten \(p\)-moment. This yields the final density-of-states bound parameterized by binary Rényi entropy. Finally, Corollary~\ref{cor:coarse} gives a simplified binary-entropy bound, and Corollary~\ref{cor:response-density} gives an interaction-dependent refinement.

\begin{lemma}\label{lem:smoothing}
	Let $\rho$ be a density operator satisfying
	\[
	\Tr\bigl((H-\lambda_0I)\rho\bigr)\le\theta t,
	\qquad 0\le\theta<1,
	\qquad t>0,
	\]
	and suppose that $\Tr(\rho^p)\le Q$ for some $p>1$ and $Q>0$. Then
	\begin{equation}\label{eq:smoothing-rank}
		N_H(\lambda_0+t)
		\ge
		(1-\theta)^{p/(p-1)}Q^{-1/(p-1)}.
	\end{equation}
\end{lemma}

\begin{proof}
	Let
	$
	P=\1[H\le \lambda_0+t]
	$
	be the spectral projector onto the low-energy subspace. Since
	$
	H-\lambda_0I\succeq t(I-P),
	$
	the energy bound implies
	\[
	\Tr\bigl((I-P)\rho\bigr)
	\le
	\frac1t\Tr\bigl((H-\lambda_0I)\rho\bigr)
	\le\theta.
	\]
	Therefore,
	$
	\Tr(P\rho)\ge1-\theta.
	$
	H\"older's inequality gives
	\begin{align*}
		\Tr(P\rho)
		&\le
		\|P\|_{p/(p-1)}\|\rho\|_p\\
		&=
		N_H(\lambda_0+t)^{(p-1)/p}\Tr(\rho^p)^{1/p}\\
		&\le
		N_H(\lambda_0+t)^{(p-1)/p}Q^{1/p}.
	\end{align*}
	Combining the two bounds and raising both sides to the power $p/(p-1)$ proves \eqref{eq:smoothing-rank}.
\end{proof}

We now derive the main result of this section: a lower bound on the density of low-energy states for $k$-local Hamiltonians.

\begin{theorem}
	\label{thm:exact}
	Let \(H\) be an \(n\)-qubit \(k\)-local Hamiltonian with total interaction
	strength \(M\). Let \(\mu>0\) and choose \(b\in(0,1/2)\) such that \(3kb<\mu\).
	Then, for every \(p>1\), 
	\begin{equation}\label{eq:exact-main}
		N_H(\lambda_0+\mu M)
		\ge
		\left(1-\frac{3kb}{\mu}\right)^{p/(p-1)}
		\exp\bigl(nH_p(b)\bigr).
	\end{equation}
\end{theorem}

\begin{proof}
	If \(M=0\), then \(H=0\) and the claim is immediate. Hence, assume \(M>0\). Let $\omega$ be a ground-state density operator and set
	\[
	\rho_b=\D_b^{\otimes n}(\omega).
	\]
	By Lemma~\ref{lem:energy},
	\[
	\Tr\bigl((H-\lambda_0I)\rho_b\bigr)
	\le
	3kbM
	=
	\frac{3kb}{\mu}\,\mu M.
	\]
	Moreover, Lemma~\ref{lem:moment} gives
	\[
	\Tr(\rho_b^p)
	\le
	\exp\bigl(-(p-1)nH_p(b)\bigr).
	\]
	Applying Lemma~\ref{lem:smoothing} with
	\[
	t=\mu M,
	\qquad
	\theta=\frac{3kb}{\mu},
	\qquad
	Q=\exp\bigl(-(p-1)nH_p(b)\bigr)
	\]
	yields
	$$
	N_H(\lambda_0+\mu M)
	\ge
	\left(1-\frac{3kb}{\mu}\right)^{p/(p-1)}Q^{-1/(p-1)}
	=
	\left(1-\frac{3kb}{\mu}\right)^{p/(p-1)}
	\exp\bigl(nH_p(b)\bigr).
	$$
\end{proof}

We next derive a simpler lower bound, expressed in terms of the binary entropy function, that enables direct comparison with the results of~\cite{beat} and~\cite{MGT26}.

\begin{corollary}\label{cor:coarse}
	Let \(H\) be an \(n\)-qubit \(k\)-local Hamiltonian with total interaction
	strength \(M\).
	For \(0<\mu\leq1\) satisfying $\mu=\Omega(1/\operatorname{poly}(n))$,
	\[
	N_H(\lambda_0+\mu M)
	=
	\Omega^*\left(
	2^{\frac35nh_2\left(\frac{\mu}{6k}\right)}
	\right).
	\]
\end{corollary}

\begin{proof}
	Set
	\[
	b=\frac{\mu}{6k},
	\qquad
	L=\log(1/b),
	\qquad
	p=1+\frac1L.
	\]
	Since $\mu\le1$ and we are interested in the regime $k\ge2$,
	\[
	0<b\le\frac1{12}<\frac12,
	\qquad
	L\ge\log12>1.
	\]
	Moreover,
	\[
	\frac{3kb}{\mu}=\frac12.
	\]
	Since $p/(p-1)=L+1$, Theorem~\ref{thm:exact} gives
	\begin{equation}\label{eq:coarse-intermediate}
		N_H(\lambda_0+\mu M)
		\ge
		2^{-(L+1)}\exp\bigl(nH_p(b)\bigr).
	\end{equation}
	
	It remains to lower-bound $H_p(b)$. The identities
	$b=\e^{-L}$ and $p=1+1/L$ imply
	\[
	b^p=\frac b\e.
	\]
	Since $L>1$, the concavity of $x\mapsto x^{1/L}$ gives
	\[
	(1-b)^{1/L}\le1-\frac bL.
	\]
	Consequently,
	$$
	(1-b)^p+b^p
	\le
	(1-b)\left(1-\frac bL\right)+\frac b\e
	=
	1-b\left(1-\e^{-1}+\frac{1-b}{L}\right).
	$$
	Using $-\log(1-x)\ge x$, we obtain
	\[
	H_p(b)
	\ge
	b\bigl((1-\e^{-1})L+1-b\bigr).
	\]
	Since $b\le1/12<\e^{-1}$,
	\[
	(1-\e^{-1})L+1-b
	\ge
	(1-\e^{-1})(L+1).
	\]
	On the other hand,
	\[
	h(b)
	=
	-b\log(b)-(1-b)\log(1-b)
	=
	bL-(1-b)\log(1-b)
	\le
	b(L+1).
	\]
	It follows that
	\[
	H_p(b)
	\ge
	(1-\e^{-1})h(b)
	\ge
	\frac35h(b).
	\]
	Substituting this into \eqref{eq:coarse-intermediate} and using
	$h_2(b)=h(b)/\log 2$ yields
	\[
	N_H(\lambda_0+\mu M)
	\ge
	2^{\frac35nh_2(b)-(L+1)}.
	\]
	Finally,
	\[
	2^{-(L+1)}
	=
	\frac12\left(\frac{\mu}{6k}\right)^{\log 2}.
	\]
	Since $k\le n$ and $\mu=\Omega(1/\operatorname{poly}(n))$, this is an inverse-polynomial factor in $n$ and is absorbed by the $\Omega^*$ notation.
\end{proof}

\begin{remark}
	Corollary~\ref{cor:coarse} is a simpler but coarser consequence of the Rényi-parameterized bound in Theorem~\ref{thm:exact}. It is obtained by choosing particular values of the depolarizing parameter \(b\) and the Rényi parameter \(p\), together with several simplifying estimates. Optimizing these parameters or sharpening the estimates may improve the constants \(3/5\) and \(6\) appearing in the exponent of Corollary~\ref{cor:coarse}.
\end{remark}

Although the constants in Corollary~\ref{cor:coarse} may be improved, its binary-entropy dependence is optimal up to constant factors in the exponent for fixed \(k\).

\begin{proposition}
	\label{prop:tightness}
	Fix \(k\geq2\) and \(0<\mu<1/2\). There exists a family of classical diagonal \(n\)-qubit \(k\)-local Hamiltonians
	\(H_{\mathrm{cl}}^{(n)}\) such that
	\[
	N_{H_{\mathrm{cl}}^{(n)}}\left(
	\lambda_0\bigl(H_{\mathrm{cl}}^{(n)}\bigr)+\mu M
	\right)
	=
	2^{\Theta\left(
	nh_2\left(\frac{\mu}{6k}\right)
	\right)}.
	\]
	Consequently, the exponent in Corollary~\ref{cor:coarse} is
	optimal up to constant factors.
\end{proposition}

\begin{proof}
	Consider the classical diagonal Hamiltonian
	\[
	H_{\mathrm{cl}}^{(n)}
	=
	\sum_{\substack{S\subseteq[n]\\ |S|=k}}
	\left(
	I_S-\ket{0^k}\!\bra{0^k}_S
	\right).
	\]
	Each interaction term acts nontrivially on exactly \(k\) qubits and has norm \(1\). Therefore,
	\[
	M=\binom{n}{k}.
	\]
	The unique ground state is \(\ket{0^n}\), with ground-state energy \(\lambda_0=0\).

	Let \(\ket{x}\) be a computational-basis state of Hamming weight \(w\). An interaction term contributes energy \(1\) precisely when its support contains at least one nonzero bit of \(x\). Hence,
	\[
	\bra{x}H_{\mathrm{cl}}^{(n)}\ket{x}
	=
	\binom{n}{k}-\binom{n-w}{k}.
	\]
	Define
	\[
	q_{k,\mu}
	\coloneqq
	1-(1-\mu)^{1/k}.
	\]
	Since \(k\) and \(\mu\) are fixed, the largest Hamming weight of an eigenstate having energy at most \(\mu M\) is
	\(q_{k,\mu}n+o(n)\). It follows from Stirling's approximation that
	\[
	N_{H_{\mathrm{cl}}^{(n)}}(\mu M)
	=
	2^{n h_2(q_{k,\mu})+o(n)}.
	\]

	Moreover,
	\[
	\frac{\mu}{k}
	\leq
	q_{k,\mu}
	\leq
	\frac{2\mu}{k}.
	\]
	Because \(h_2\) is increasing on \([0,1/2]\) and \(x\mapsto h_2(x)/x\) is decreasing, this implies
	\[
	h_2(q_{k,\mu})
	=
	\Theta\left(
	h_2\left(\frac{\mu}{6k}\right)
	\right).
	\]
	Therefore,
	\[
	N_{H_{\mathrm{cl}}^{(n)}}(\mu M)
	=
	2^{\Theta\left(
	nh_2\left(\frac{\mu}{6k}\right)
	\right)},
	\]
	which proves the claim.
\end{proof}

We now take into account the type of interactions in $H$. Substituting Lemma~\ref{cor:response-energy} for Lemma~\ref{lem:energy} in the proofs of Theorem~\ref{thm:exact} and Corollary~\ref{cor:coarse} yields the following:

\begin{corollary}
\label{cor:response-density}
Let \(H\) be an \(n\)-qubit \(k\)-local Hamiltonian with \(M>0\), and write $y=y(H)$ as defined in Eq.~\eqref{eq:interaction-response-parameter}.
\begin{enumerate}
\item[(i)] Let $\mu>0$ and choose $b\in(0,1/2)$ such that $yb<\mu$. Then for every $p>1$,
\[
 N_H(\lambda_0+\mu M)
 \ \ge\ \left(1-\frac{yb}{\mu}\right)^{p/(p-1)}\exp\!\bigl(nH_p(b)\bigr).
\]
\item[(ii)] Suppose in addition that $y\ge2$. Then for every $0<\mu\le1$ satisfying $\mu=\Omega(1/\operatorname{poly}(n))$,
\[
 N_H(\lambda_0+\mu M)
 \ =\ \Omega^{*}\!\left(2^{\frac35 n\,h_2\left(\frac{\mu}{2y}\right)}\right).
\]
\end{enumerate}
\end{corollary}

\begin{proof}
Part~(i) is the proof of Theorem~\ref{thm:exact} verbatim, with Lemma~\ref{cor:response-energy} in place of Lemma~\ref{lem:energy}: one applies Lemma~\ref{lem:smoothing} with $t=\mu M$, $\theta=yb/\mu$ and $Q=\exp(-(p-1)nH_p(b))$.

For part~(ii), set
\[
 b=\frac{\mu}{2y},\qquad L=\log(1/b),\qquad p=1+\frac1L,
\]
so that $yb/\mu=1/2$ and $p/(p-1)=L+1$. Since $\mu\le1$ and $y\ge2$ we have $0<b\le1/4<\mathrm e^{-1}$, and hence $L>1$ and $b<1/2$. These are the only properties of $b$ used in the proof of Corollary~\ref{cor:coarse}, which therefore gives $H_p(b)\ge(1-\mathrm e^{-1})h(b)\ge\tfrac35h(b)$. Substituting into part~(i) and using $h_2(b)=h(b)/\log2$,
\[
 N_H(\lambda_0+\mu M)
 \ \ge\ 2^{-(L+1)}\exp\!\bigl(nH_p(b)\bigr)
 \ \ge\ \frac12\,b^{\log2}\;2^{\frac35 n\,h_2\left(\frac{\mu}{2y}\right)} .
\]
Finally, since $y\le3n$, we have
\[
 b=\frac{\mu}{2y}\ge\frac{\mu}{6n}
 =\Omega(1/\operatorname{poly}(n)).
\]
Thus the prefactor $\tfrac12 b^{\log 2}$ is bounded below by an inverse-polynomial factor in $n$ and is absorbed by the $\Omega^*$ notation.
\end{proof}

%% file: algorithm.tex
\section{Algorithmic Implications}\label{sec:algorithmic-implications}

We now combine our density-of-states bound in Corollary~\ref{cor:coarse} with the quantum algorithm of~\cite[Section~3]{MGT26}, which is based on the low-energy state-preparation procedure of~\cite[Section~6]{lintong}. 

Let
\[
\widetilde H=H\otimes I_{\mathrm{anc}}
\]
be the extended Hamiltonian acting as \(H\) on an \(n\)-qubit system register and as the identity on an \(n\)-qubit ancilla register. We choose the maximally entangled state
\[
\ket{\Phi}
=
\frac{1}{\sqrt{2^n}}
\sum_{z\in\{0,1\}^n}
\ket{z}_{\mathrm{sys}}\ket{z}_{\mathrm{anc}}
\]
as the initial state. Recall that $\ket{\Phi}$ can be prepared efficiently.

\begin{lemma}\label{lem:maximally-entangled-overlap}
	For every \(E\in\mathbb R\), the overlap of \(\ket{\Phi}\) with the
	low-energy subspace of \(\widetilde H\) corresponding to \(P_{\le E}\) is
	\[
	\bra{\Phi}
	\bigl(P_{\leq E}\otimes I_{\mathrm{anc}}\bigr)
	\ket{\Phi}
	=
	\frac{N_H(E)}{2^n}.
	\]
\end{lemma}

\begin{proof}
	The reduced density operator of \(\ket{\Phi}\) on the system register is \(I/2^n\). Therefore,
	$$
	\bra{\Phi}
	\bigl(P_{\leq E}\otimes I_{\mathrm{anc}}\bigr)
	\ket{\Phi}
	=
	\operatorname{Tr}\left(
	P_{\leq E}\frac{I}{2^n}
	\right)
	=
	\frac{\operatorname{Tr}(P_{\leq E})}{2^n}
	=
	\frac{N_H(E)}{2^n}.
	$$
\end{proof}

We can now prove Theorem~\ref{th:main}.

\begin{theorem*}[Restatement of Theorem~\ref{th:main}]
	Let \(H\) be an \(n\)-qubit \(k\)-local Hamiltonian. For any
	\(\varepsilon=\Omega(1/\operatorname{poly}(n))\), there exists a quantum algorithm that, with high probability, prepares an \(n\)-qubit quantum state \(\rho\) satisfying
	\[
	\operatorname{Tr}[H\rho]
	\leq
	\lambda_0+\varepsilon M
	\]
	and outputs a corresponding energy estimate \(\widehat{E}\) lying within the window
	\[
	\widehat{E}
	\in
	[\lambda_0,\lambda_0+\varepsilon M].
	\]
	The running time of the algorithm is
	\[
	O^*\left(
	2^{\frac{n}{2}\left(
		1-\frac{3}{5}
		h_2\left(\frac{\varepsilon}{6k}\right)
		\right)}
	\right),
	\]
	where $h_2(\cdot)$ denotes the binary entropy function in bits.
\end{theorem*}

\begin{proof}
	If \(M=0\), the statement is immediate. Otherwise, set
	\[
	\mu=\left(1-\frac1n\right)\varepsilon.
	\]
	By Lemma~\ref{lem:maximally-entangled-overlap} and Corollary~\ref{cor:coarse}, the overlap \(\gamma\) of \(\ket{\Phi}\) with the subspace of \(\widetilde H\) having energies at most \(\lambda_0+\mu M\) satisfies
	\begin{equation}\label{eq:gamma}
		\gamma
		=
		\frac{N_H(\lambda_0+\mu M)}{2^n}
		=
		\Omega^*\left(
		2^{-n+\frac35nh_2\left(\frac{\mu}{6k}\right)}
		\right).
	\end{equation}	
	The remainder of the proof follows by the same argument as the proof of Theorem~2 in~\cite{MGT26}, with \(E_d\) replaced by \(\lambda_0\) and the overlap bound used there replaced by Equation~\eqref{eq:gamma}. We therefore omit the repeated details.
\end{proof}

\subsection{Runtime Bounds for Restricted Interaction Types}

We now use the parameter $y(H)$ to obtain runtime bounds for specific interaction models on arbitrary graphs. We first state the runtime in terms of $y(H)$ and then evaluate this parameter for the models considered below.

\begin{corollary}
\label{cor:response-algorithms}
Let $H$ be an $n$-qubit $k$-local Hamiltonian with $M>0$
and $y=y(H)\ge2$. For any $0<\varepsilon<1$ satisfying $\varepsilon=\Omega(1/\operatorname{poly}(n))$, there exists a quantum algorithm that, with high probability, prepares an $n$-qubit state $\rho$ satisfying
\[
 \operatorname{Tr}(H\rho)\le\lambda_0+\varepsilon M
\]
and outputs an estimate $\widehat E\in[\lambda_0,\lambda_0+\varepsilon M]$,
in time
\[
 O^*\!\left(    
 2^{\frac n2\left(1-\frac35h_2\left(\frac{\varepsilon}{2y}\right)\right)}
 \right).
\]
\end{corollary}

\begin{proof}
Set $\mu=(1-1/n)\varepsilon$. Since $y\ge2$ and $\mu=\Omega(1/\operatorname{poly}(n))$, Lemma~\ref{lem:maximally-entangled-overlap} and Corollary~\ref{cor:response-density}~(ii) give, for the overlap $\gamma$ appearing in the proof of Theorem~\ref{th:main},
$$\gamma
=
		\frac{N_H(\lambda_0+\mu M)}{2^n}
		=
\Omega^*\left(
		2^{-n+\frac35nh_2\left(\frac{\mu}{2y}\right)}
		\right).$$
The proof of Theorem~\ref{th:main} now applies verbatim, with this bound for $\gamma$ in place of the one supplied there by Corollary~\ref{cor:coarse}.
\end{proof}

We apply Corollary~\ref{cor:response-algorithms} to the following models. In each model below, $H$ is a sum of local terms acting on one or two qubits. The two-qubit terms are indexed by unordered pairs $\{i,j\}$; the pairs carrying a nonzero coupling form the interaction graph of $H$, about which we make no assumption whatsoever. All coupling coefficients are real, and we write $\sum_{i<j}$ for the sum over interacting pairs. We assume $M>0$ throughout.

\paragraph{Heisenberg antiferromagnet.}
The {Heisenberg antiferromagnet} has Hamiltonian
\[
 H_{\mathrm{Heis}}
 =\sum_{i<j}J_{ij}(X_iX_j+Y_iY_j+Z_iZ_j),
 \qquad J_{ij}\ge0.
\]
Each local term satisfies $\|h_{ij}\|=3J_{ij}$. Further,
$R_i(h_{ij})=R_j(h_{ij})=0$ due to Eq.~\eqref{eq:Rsplit}. Hence
\[
 M=3\sum_{i<j}J_{ij},\qquad
 y(H_{\mathrm{Heis}})
 =\frac{4}{M}\sum_{i<j}\|h_{ij}\|=4.
\]

\paragraph{XY model.}
The $XY$ model has Hamiltonian
\[
 H=\sum_{i<j}J_{ij}(X_iX_j+Y_iY_j).
\]
Each edge term satisfies $\|h_{ij}\|=2|J_{ij}|$ and $R_i(h_{ij})=R_j(h_{ij})=0$. Hence
\[
 M=2\sum_{i<j}|J_{ij}|,\qquad
 y(H)=\frac{4}{M}\sum_{i<j}\|h_{ij}\|=4.
\]

\paragraph{Transverse-field Ising model.}
The transverse-field Ising model has Hamiltonian
\[
 H=\sum_{i<j}J_{ij}Z_iZ_j+\sum_i g_iX_i.
\]
Write $h_{ij}=J_{ij}Z_iZ_j$ and $h_i=g_iX_i$, and set $W=\sum_{i<j}|J_{ij}|$ and $B=\sum_i|g_i|$. The local norms are $\|h_{ij}\|=|J_{ij}|$ and $\|h_i\|=|g_i|$, while $R_i(h_{ij})=R_j(h_{ij})=R_i(h_i)=0$. Hence
\[
 M=W+B,\qquad
 y(H)=\frac{2(2W+B)}{W+B}
     =\frac{4W+2B}{W+B}\in[2,4].
\]

\paragraph{Classical Ising model.}
The classical Ising model has Hamiltonian
\[
 H=\sum_{i<j}J_{ij}Z_iZ_j.
\]
Each edge term satisfies $\|h_{ij}\|=|J_{ij}|$ and $R_i(h_{ij})=R_j(h_{ij})=0$. Hence
\[
 M=\sum_{i<j}|J_{ij}|,\qquad
 y(H)=\frac{4}{M}\sum_{i<j}\|h_{ij}\|=4.
\]
This is the case $B=0$ of the transverse-field Ising model. Since $H$ is diagonal in the computational basis, a ground state can be chosen to be a computational basis state. For classical Hamiltonians such as this one the depolarizing argument is not
optimal. Since the ground state is a computational basis state, the low-energy states can instead be counted directly, by perturbing an optimal assignment; this gives a stronger bound than
Corollary~\ref{cor:response-algorithms}. It is the approach taken for \textsc{max-$k$-sat} in~\cite{classicalMaxKSAT}.

\paragraph{Comparison of running times.}
For the models above, write the running-time bound as
\[
 O^*\!\left(2^{\frac n2(1-\frac35h_2(a))}\right),
 \qquad a:=\frac{\varepsilon}{2y},
\]
and call $a$ the \emph{entropy argument}. Since $0<a<1/4$ and $h_2$ is increasing on this interval, a larger entropy argument gives a smaller runtime exponent. Table~\ref{tab:response} compares these arguments with the general bound in Theorem~\ref{th:main}.

\begin{table}[htbp]
\centering
\renewcommand{\arraystretch}{1.25}
\begin{tabular}{@{}llcc@{}}
\hline
Model & $M$ & $y(H)$ & Entropy argument \\
\hline
General $k$-local
 & $\sum_\alpha\|h_\alpha\|$
 & $\le3k$
 & $\varepsilon/(6k)$ \\
Heisenberg antiferromagnet
 & $3\sum_{i<j}J_{ij}$
 & $4$
 & $\varepsilon/8$ \\
$XY$ model
 & $2\sum_{i<j}|J_{ij}|$
 & $4$
 & $\varepsilon/8$ \\
Transverse-field Ising
 & $W+B$
 & $\dfrac{4W+2B}{W+B}$
 & $\dfrac{\varepsilon(W+B)}{8W+4B}$ \\
Classical Ising
 & $\sum_{i<j}|J_{ij}|$
 & $4$
 & $\varepsilon/8$ \\
\hline
\end{tabular}
\caption{Comparison of entropy arguments for arbitrary interaction graphs. The first row gives Theorem~\ref{th:main}, with argument $\varepsilon/12$ when $k=2$; the remaining rows follow from Corollary~\ref{cor:response-algorithms}. The runtime bounds use error $\varepsilon M$, with $M$ as displayed.}
\label{tab:response}
\end{table}

%% file: conclusion.tex
\FloatBarrier
\section{Conclusion}\label{sec:conclusion}
We have shown that every \(k\)-local Hamiltonian possesses a large low-energy subspace close to its ground-state energy governed by the binary-entropy function. More precisely, our density-of-states bound in Corollary~\ref{cor:coarse} gives
\[
N_H(\lambda_0+\mu M)
=
\Omega^*\left(
2^{\frac35nh_2\left(\frac{\mu}{6k}\right)}
\right).
\]
Moreover, Proposition~\ref{prop:tightness} shows that, for fixed \(k\), this bound is optimal up to constant factors in the exponent, as witnessed by a classical diagonal Hamiltonian. Thus, the binary-entropy dependence captures the correct worst-case scaling of the low-energy density of states.

Combining this bound with the algorithmic framework of~\cite{MGT26} gives quantum algorithms for low-energy state preparation and estimation with running time
\[
O^*\left(
2^{\frac n2\left(
	1-\frac35h_2\left(\frac{\varepsilon}{6k}\right)
	\right)}
\right).
\]
The algorithms directly improve upon the runtimes established in~\cite{beat}. The parameter $y(H)$ further tightens these bounds for the models summarized in Table~\ref{tab:response}.

Our result provides a worst-case guarantee for arbitrary $k$-local Hamiltonians and does not require geometric locality, a variational ansatz, or structural assumptions on the ground state. Nevertheless, several limitations should be emphasized. The running time remains exponential, and the output low-energy state $\rho$ does not guarantee a high fidelity with the ground state.\footnote{Low expected energy implies large ground-space overlap only when the spectral gap is sufficiently large. Therefore, our results imply large ground-space overlap when \(\varepsilon M\ll\Delta\), where \(\Delta\) denotes the spectral gap.} Moreover, the accuracy is measured relative to $M$. For dense interaction graphs, \(M\) may grow with the number of qubits $n$. 

Buhrman et al.~\cite{beat} observed that GAP-ETH~\cite{GAPETH1,GAPETH2} asserts the existence of constants \(c,\varepsilon>0\) such that approximating the ground-state energy of a \(3\)-local Hamiltonian within error \(\pm\varepsilon M\) cannot be done classically in time \(O(2^{cn})\). They accordingly suggested a \emph{quantum} GAP-ETH under which no quantum algorithm could solve this problem in time \(O^*(2^{cn/2})\). In this sense, the runtime in Theorem~\ref{th:main} updates the quantitative benchmark provided by Theorem~1 of~\cite{beat} for the proposed quantum GAP-ETH. Moreover, our running time remains exponential, so the result is consistent with quantum GAP-ETH. Establishing matching conditional lower bounds remains an interesting open problem.

In summary, our results show that low-energy state preparation and estimation for general $k$-local Hamiltonians are even more accessible than previous results suggested.

%% file: appendix.tex
\section{Properties of the Qubit Depolarizing Channel}\label{ap: properties}

\begin{lemma}[Proof of Eq.~\eqref{eq:pauli-representation}]\label{alem:qubit_pauli}
	Let $X,Y,Z$ denote the Pauli operators. Then $\D_b$ admits the random-unitary representation
	\begin{equation}\label{aeq:pauli-representation}
		\D_b(\rho)
		=
		\left(1-\frac{3b}{2}\right)\rho
		+\frac b2\left(X\rho X+Y\rho Y+Z\rho Z\right).
	\end{equation}
\end{lemma}

\begin{proof}
	The Pauli operators satisfy
	\[
	\rho+X\rho X+Y\rho Y+Z\rho Z=2I_2.
	\]
	Therefore,
	\[
	X\rho X+Y\rho Y+Z\rho Z=2I_2-\rho.
	\]
	Substituting this into \eqref{aeq:pauli-representation} gives
	$$
	\left(1-\frac{3b}{2}\right)\rho
	+\frac b2(2I_2-\rho)
	=(1-2b)\rho+bI_2
	=\D_b(\rho).
	$$
\end{proof}

\begin{lemma}\label{lem:qubit_cptp}
	For $0\le b\le2/3$, the map $\D_b$ is CPTP.
\end{lemma}

\begin{proof}
	By Lemma~\ref{alem:qubit_pauli},
	\[
	\D_b(\rho)
	=
	\left(1-\frac{3b}{2}\right)\rho
	+\frac b2\left(X\rho X+Y\rho Y+Z\rho Z\right).
	\]
	The four coefficients are non-negative for $0\le b\le2/3$ and sum to $1$. Hence $\D_b$ is a convex combination of unitary channels and is CPTP.
\end{proof}